\documentclass[jmp,10pt,superscriptaddress,notitlepage,aps,nofootinbib]{revtex4-1}

\usepackage{amsmath}
\usepackage{amsthm}
\usepackage{mathtools}
\usepackage{amssymb}
\usepackage{xspace}
\usepackage{hyperref}
\usepackage{multirow}
\usepackage{tikz}
\usetikzlibrary{calc,matrix}

\newlength{\ww}
\newlength{\www}
\newcommand\w[1]{\makebox[\ww]{$#1$}}

\newtheorem{theorem}{Theorem}[section]
\newtheorem{lemma}[theorem]{Lemma}
\newtheorem{proposition}[theorem]{Proposition}
\newtheorem{remark}[theorem]{Remark}
\newtheorem{definition}[theorem]{Definition}
\newtheorem{corollary}[theorem]{Corollary}

\def\idty{\mathbbm{1}} 

\def\EF{\Phi} 
\def\dr{m} 
\def\We{W_{\EF}}  
\def\kk{i}   
\def\norm #1{\Vert #1\Vert}
\def\Norm #1{\left\Vert #1\right\Vert}
\def\abs#1{\vert#1\vert}
\def\absbig#1{\left\vert#1\right\vert}

\def\scp#1#2{\langle#1,#2\rangle}                                
\def\scpbig#1#2{\left\langle#1,#2\right\rangle}                                
\def\HH{\mathcal H}
\def\capconst{K}    
\def\smallconst{c}    

\def\reals{\mathbb{R}}                             
\def\integers{\mathbb{Z}}                             
\def\complex{\mathbb{C}}                           
\def\torus{\mathbb{T}}                              

\def\re{\Re e\,}
\def\im{\Im m\,}

\def\DC{D\!C}   
\def\trans{T}
\newcommand{\transz}[2][z]{\trans_{#2}(#1)}
\def\lyap{\gamma}
\def\inv{{-1}}
\DeclareMathOperator\tr{tr\,}                                   
\def\pauli{\sigma}
\DeclareMathOperator\mes{mes}
\DeclareMathOperator\diag{diag}

\def\Pll#1{P_{\leq#1}}

\def\Pgg#1{P_{>#1}}

\def\un{unitary\xspace}

\def\qw{quantum walk\xspace}
\def\qws{quantum walks\xspace}

\def\wrt{with respect to\xspace}

\usepackage[T1]{fontenc}
\usepackage[utf8x]{inputenc}
\usepackage{bbm}

\usepackage{pgfplots}
\usepackage{tikz}
\usepackage{tikzscale}
\usepackage{standalone}
\usetikzlibrary{arrows, positioning, calc,decorations.pathreplacing}
\pgfplotsset{compat=newest}

\begin{document}

\title{Anderson localization for electric quantum walks and skew-shift CMV matrices}

\author{C. Cedzich}
\affiliation{Laboratoire de Recherche en Informatique (LRI), Universit\'{e} Paris Sud, CNRS, Centrale Sup\'{e}lec, \\B\^{a}t. 650, Rue Noetzlin, 91190 Gif-sur-Yvette, France}
\affiliation{Institut f\"ur Theoretische Physik, Leibniz Universit\"at Hannover, Appelstr. 2, 30167 Hannover, Germany}
\author{A.~H. Werner}
\affiliation{{QMATH}, Department of Mathematical Sciences, University of Copenhagen, Universitetsparken 5, 2100 Copenhagen, Denmark,}
\affiliation{{NBIA}, Niels Bohr Institute, University of Copenhagen, Denmark}

\begin{abstract}
We consider the spectral and dynamical properties of one-dimensional quantum walks placed into homogenous electric fields according to a discrete version of the minimal coupling principle. We show that for all irrational fields the absolutely continuous spectrum of these systems is empty, and prove Anderson localization for almost all (irrational) fields. This result closes a gap which was left open in the original study of electric quantum walks: a spectral and dynamical characterization of these systems for typical fields. Additionally, we derive an analytic and explicit expression for the Lyapunov exponent of this model. Making use of a connection between quantum walks and CMV matrices our result implies Anderson localization for CMV matrices with a particular choice of skew-shift Verblunsky coefficients as well as for quasi-periodic unitary band matrices.
\end{abstract}

\maketitle

\section{Introduction}
Time-discrete quantum walks have recently gained a lot of attention from very different points of view as a model in computer science, quantum physics and mathematics: Considered as the quantum evolution of a single particle with internal degree of freedom on a lattice or graph in discrete time-steps and with bounded hopping length, they can serve as the basis for single particle quantum simulators. In this context, quantum walks have been shown to capture many single and few particle quantum effects such as ballistic transport \cite{ambainis2001one,Grimmet,TimeRandom}, decoherence \cite{SpacetimeRandom,TimeRandom,schreiber2011decoherence}, dynamical localization \cite{SpaceRandom,Joye_Merkli,joye_d_dim_loc} and the formation of bound states \cite{ahlbrecht2012molecular, schreiber20122d, sansoni2012two} both with regards to theoretical as well as experimental physics in diverse architectures \cite{karski2009quantum,zahringer2010realization,schmitz2009quantum,PhysRevLett.104.050502,PhysRevA.72.062317,topo_exp}. More recently, quantum walks have been shown to provide a testbed for symmetry protected topological order where the corresponding invariants can be shown to provide a complete topological classification without assumptions on translation invariance \cite{UsOnTop_short,UsOnTop_long,UsOnTI,WeAreSchur}.

Complementary to this quantum simulation point of view, quantum walks can be seen as a generalization of classical random walks to the quantum regime. Here, the increased ballistic spreading behaviour as compared to classical diffusion has interesting algorithmic applications which include for example search algorithms, element distinctness, quantum information processing and applications to the graph isomorphism problem \cite{childs2003exponential, ambainis2003quantum, childs2009universal,berry2011two,portugal2018quantum}. From a practical point of view it is hence important to ascertain how experimental imperfections might change the performance of these algorithms. Results with regard to spatial and temporal fluctuations in the coin parameters have been obtained in a number of recent papers, showing decoherence effects that imply a transition from ballistic to diffusive spreading for temporal fluctuations \cite{TimeRandom,Joye2011}, Anderson or even dynamical localization in the case of time-independent disorder \cite{SpaceRandom,Joye_Merkli,joye_d_dim_loc}, as well as diffusive spreading if both types of disorder are present \cite{SpacetimeRandom}.

In the following, we will be concerned with similar questions in the quasi-periodic regime. In line with the investigation of simulable physical effects in discrete time, electric and magnetic fields were recently introduced to the quantum walk setup using a discrete analogue of the minimal coupling principle \cite{ewalks,UsOnMag}. It has been found that an external discrete electric field changes the spectral properties and the dynamical behaviour of a given one-dimensional quantum walk dramatically \cite{ewalks}. Whereas for rational fields Bloch-like oscillations are observed on short time scales before the ballistic behaviour dominates eventually, the case of irrational fields is more involved: it was shown that irrational fields which are extremely well approximable by rational ones lead to purely singular continuous spectrum and hierarchical motion whereas for badly approximable fields like the Golden Ratio numerical studies suggest Anderson localization, see the summary in Table \ref{tab:ewalks}.

Yet, the important question about the generic behaviour has been left open: while the set of rational, extremely well approximable and badly approximable fields each constitute a dense subset in the set of all fields they all have measure zero. This gap is filled in the present work in which we show that Lebesgue-typical fields lead to Anderson localization, i.e. pure point spectrum with exponentially decaying eigenfunctions.

\begin{table}[h]
\begin{center}
\begin{tabular}{c||c|c|c||c}
       &rational&\ almost rational&very irrational&almost all\\\hline
  cont. fract.&\multirow{2}{*}{terminates}  & $c_i\to\infty$ & \multirow{2}{*}{$c_i$ bounded} & \\
  expansion &           &   rapidly & \\\hline
  \multirow{2}{*}{propagation}& ballistic&\multirow{2}{*}{hierarchical} & \multirow{2}{*}{localized}&Anderson\\
         &w. revivals    & & & localization   \\\hline
         \multirow{2}{*}{$\sigma(\We)$}    &   \multirow{2}{*}{$\sigma_{ac}(\We)$}    &   \multirow{2}{*}{$\sigma_{sc}(\We)$}    &   \multirow{2}{*}{$\sigma_{pp}(\We)$}  &   \multirow{2}{*}{$\sigma_{pp}(\We)$}      \\
         &&& \\\hline
   \multirow{2}{*}{status} & \multirow{2}{*}{proved}   &\multirow{2}{*}{proved} & \multirow{2}{*}{num. evidence}&\multirow{2}{*}{proved}\\ &&&
\end{tabular}
\caption{Overview on the connection between the properties of the electric field $\EF/(2\pi)$, its \index{continued fraction}continued fraction coefficients $c_i$, and the propagation and the spectrum of the electric walk $\We$. The results for rational, almost rational and very irrational fields are proved in \cite{ewalks}. The main objective of this manuscript is to prove the result in the last column where fields do not admit a characterization in terms of $c_i$.}
\label{tab:ewalks}
\end{center}
\end{table}

Proving localization statements for quasi-periodic systems has attracted a lot of attention over the last decades. At the center of this attention is the almost Mathieu operator, which is a tight-binding Hamiltonian with cosine potential and variable coupling constant whose spectral and dynamical properties depend sensitively on the parameter choices. For example, Anderson localization was proved for this operator (or generalizations thereof) at large coupling and Diophantine frequencies for almost all offsets in \cite{FroehlichMSAquasi,Sinai1987}, and non-perturbatively in \cite{bourgain_goldstein}. In contrast, in the critical case where the operator describes the motion of a single particle in a perpendicular magnetic field the spectrum is almost surely purely singular continuous \cite{gordon1997}.

The unitary analogue of the critical almost Mathieu operator introduced in \cite{Linden2009} is a one-dimensional shift-coin quantum walk with quasi-periodic coin. In close analogy to the self-adjoint case, this quantum walk has purely singular continuous spectrum for all irrational frequencies and almost all offsets \cite{InhomogeneousWalkFillman}. From a physics point of view this model describes the discrete evolution of a quantum mechanical particle on a two-dimensional lattice under the influence of a discrete homogeneous magnetic field \cite{UsOnMag,UsOnCantor}.

Coined quantum walks are closely related to the subclass of doubly-infinite ``sparse'' CMV matrices for which every second Verblunsky coefficient vanishes \cite{CMV,CGMV}. Recently, CMV matrics with quasi-periodic Verblunksy coefficients were shown to obey Anderson localization for almost all irrational frequencies under the assumption of positivity of the Lyapunov exponent \cite{wang_damanik}. However, due to the nature of the unitary equivalence between quantum walks and CMV matrices the quasi-periodicity of the coins does not automatically translate to quasi-periodicity of the Verblunsky coefficients. A class of quantum walks for which this association fails is the electric walks studied in this paper. Accordingly, our results imply Anderson localization for sparse CMV matrices where the non-vanishing Verblunsky coefficients are generated by a two-dimensional skew-shift. Using a technique known as ``sieving'' this proves Anderson localization for also for the full CMV matrices.

Another type of models closely related to the systems in this paper is the unitary band matrices studied in \cite{AlainUnitaryBandMats} which contain CMV matrices as a subset. The authors consider quasi-periodic models and show that for Liouville frequencies the spectrum is purely singular continuous, which is in accordance with the findings in \cite{ewalks}. We complement their results and show that for almost all frequencies quasi-periodic unitary band matricies obey Anderson localization.

The paper is organized as follows: in Section \ref{sec:system} we define the system under consideration, state the main result in Theorem \ref{thm:loclEQW} and discuss its implications in connection with previous work. This proof of the theorem is subsequently given in Sections \ref{sec:no_ac} and \ref{sec:proof_loc}.

\section{System}\label{sec:system}
\subsection{The physical model}
Quantum walks describe the time-discrete evolution of a single particle with an internal degree of freedom on a lattice under the additional assumption of a finite propagation speed. In this paper we consider particles on the one-dimensional lattice $\integers$ with two-dimensional internal degree of freedom which fixes the Hilbert space as $\HH=\ell_2(\integers)\otimes \complex^2$. The unitary timestep operators are \emph{shift-coin quantum walks} given as a product
\begin{align}\label{eq:defShCoin}
  W = CS,
\end{align}
where $C=\bigoplus_\integers C_x$ is called the coin operator with $C_x\in U(2)$ acting only on the internal degree of freedom and $S$ is a conditional shift operator acting on basis states of $\HH$ as $S(\delta_x\otimes e_\pm) = \delta_{x\pm 1}\otimes e_\pm$. In the translation invariant case the coin $C$ acts the same everywhere, i.e. $C=\idty\otimes C_0$ with
\begin{equation}\label{eq:coin_TI}
  C_0=e^{i\eta}\begin{pmatrix} a & b \\ -b^*  & a^* \end{pmatrix},\;\quad \abs{a}^2+\abs{b}^2=1,
\end{equation}
and the walk $W$ generically has purely absolutely continuous spectrum and exhibits ballistic transport \cite{Grimmet,TimeRandom}. In contrast, if the coin is given by an i.i.d. random function $\Omega\ni\omega\mapsto C_\omega$ the spectrum of $W$ is pure point and the walk exhibits dynamical localization \cite{SpaceRandom,bucaj2019localization}.

Recently, it has been studied how discrete electromagnetic fields can be introduced in quantum walk systems by an approach similar to minimal coupling in continuous time \cite{UsOnMag}. In the one-dimensional setting considered here, homogeneous electric fields are introduced via the modification
\begin{align}\label{eq:EQW}
  \We = e^{i (\EF Q + \theta)} W,
\end{align}
where $W$ is the shift-coin quantum walk from \eqref{eq:defShCoin} and $Q$ denotes the position operator $Q(\delta_x\otimes e_\pm) = x\,\delta_x\otimes e_\pm$. As a standing assuption, we take $W$ to be translation invariant throughout this paper. Occasionally we shall write $\We(\theta)$ to make explicit the $\theta$-dependence of the electric walk model. With regards to the spectrum however, it is apparent from \eqref{eq:EQW} that for all $\theta\in[0,2\pi]$
\begin{equation}\label{eq:thetashift}
  z\in\sigma(\We(\theta))\quad\Rightarrow\quad e^{-i\theta}z\in\sigma(\We(0)).
\end{equation}

In accordance with \cite{ewalks} and \cite{UsOnMag} the parameter  $\EF\in[0,2\pi]$ is called the discrete electric field and $\theta\in[0,2\pi]$ describes an arbitrary offset and plays the role of the random parameter in the quasi-periodic operator $\We$. Thus, a homogenous electric field in this setting can be seen as a position dependent linear phase factor applied in each time step. Note that in \eqref{eq:EQW} a particular gauge is chosen which assures that $\We$ is time-independent and therefore makes a discussion of its spectral properties meaningful. In \cite{ewalks} it has been shown that the spectral as well as the dynamical properties of $\We$ depend sensitively on the rationality of $\EF/(2\pi)$: for rational fields the system remains translation invariant after grouping lattice sites together, which implies absolutely continuous spectrum and eventually ballistic transport. Notably, on short timescales of the order of the denominator of $\EF/(2\pi)$ revivals of the initial state occur which are exponentially sharp in this denominator. For irrational fields the picture is more involved: if $\EF/(2\pi)$ is extremely well approximable in terms of its continued fraction approximation - an example being the Liouville numbers - the dynamics of $\We$ is hierarchical in the sense that there is an infinite sequence of exponentially sharper and sharper revivals which alternate with farther and farther excursions. This type of dynamics implies that the spectrum for such fields is purely singular continuous. On the other hand, if $\EF/(2\pi)$ is badly approximable in the sense that the sequence of its continued fraction coefficients is bounded numerical evidence leads to conjecture Anderson localization. Yet, each of these sets of fields is of measure zero. The main goal of this paper is to study the case of Lebesgue-typical fields $\EF$.

\subsection{Transfer matrices}

In order to get a handle on the spectral properties and eigenfunctions of the quasi-periodic electric walk \eqref{eq:EQW} we use a transfer matrix approach that has already proved useful in the disordered setting \cite{SpaceRandom}. For notational convenience in the rest of this paper we identify $\ell_2(\integers)\otimes\complex^2\to\ell_2(\integers)$ via
\begin{equation*}
  \delta_x\otimes e_+\mapsto\delta_{2x},\qquad \delta_x\otimes e_-\mapsto\delta_{2x+1},
\end{equation*}
where $x\in\integers$ labels the lattice sites.
Then, given a shift-coin quantum walk of the form \eqref{eq:defShCoin} any solution $\phi$ to the generalized eigenvalue equation
\begin{align}\label{eq:EWeq}
  W\phi = z\phi
\end{align}
will satisfy the relation \cite{SpaceRandom}
\begin{align*}
  \begin{pmatrix}
    \phi_{2x+1}\\
    \phi_{2x+2}
    \end{pmatrix}=
 \transz{x}\begin{pmatrix}
    \phi_{2x-1}\\
    \phi_{2x}
  \end{pmatrix}\;,
  \end{align*}
  where we introduced the transfer matrix
  \begin{align}\label{eq:transfMat}
  \transz{x} =  \frac{1}{a_x}\begin{pmatrix}
    \frac{\det(C_x)}{z} & c_x\\
    -b_x & z
  \end{pmatrix} \;,
\end{align}
with the coin matrix $C_x$ at lattice site $x$ parametrized by $a_x,b_x,c_x$ and $d_x$. In other words, two consecutive components of $\phi$ satisfying \eqref{eq:EWeq} determine the whole generalized eigenvector by repeatedly applying the transfer matrix. Note that, since the shift-coin walk \eqref{eq:defShCoin} is a finite-difference operator its spectrum is characterized as follows \cite{berezanskiui1968expansions,AlainUnitaryBandMats,DamanikUniformHyperbolicity}:
\begin{lemma}\label{lem:gen_eig}
Let $\psi$ be a solution of \eqref{eq:EWeq} for some $z\in\complex$. Then $\psi$ cannot be polynomially bounded if $z\not\in\sigma(W)$. Vice versa, $\sigma(W)$ is the closure of the set of generalized eigenvalues.
\end{lemma}

In the electric walk \eqref{eq:EQW} the operator $e^{i (\EF Q + \theta)}$ implementing the electric field acts locally, wherefore we can interpret it as an additional coin operator and write $\We$ as a walk with a coin that is a quasi-periodic function of the position, i.e.
\begin{align}\label{eq:EwalkII}
 \We = \left(\bigoplus_{x\in \integers} C_0 e^{i(x \EF + \theta)}\right)\cdot S.
\end{align}
As remarked above, the random offset $e^{i\theta}$ is a global phase factor that merely shifts the spectrum of $\We$. This allows us to absorb the additional phase factor $e^{i\eta}$ coming from the determinant of $C_0$ into the offset $\theta$ and restrict our attention without loss of generality to $C_0\in SU(2)$, i.e. set $\eta=0$ in \eqref{eq:coin_TI}.

Plugging the coin in \eqref{eq:EwalkII} into \eqref{eq:transfMat}, the electric transfer matrix at lattice point $x$ is given by
  \begin{equation}\label{eq:transMatEF}
  \trans_x(\theta,z)=\trans(\tau_\EF^x(\theta),z)=\frac{1}{a}\begin{pmatrix} z^{-1} e^{i (x\EF + \theta)}&-b^*\\-b &z e^{-i(x\EF +\theta)}\end{pmatrix}.
\end{equation}
Here, the quasi-periodic shift defined by
\begin{equation*}
  \tau_\EF:\torus\to\torus,\quad\theta\mapsto\tau_\EF(\theta)=\EF+\theta,
\end{equation*}
is ergodic whenever $\EF/(2\pi)$ is irrational. Clearly, the transfer matrices only depend on the combined parameter $z^\inv e^{i\theta}$. Accordingly, if $z$ lies on the unit circle, averaging over the spectral parameter $z$ is the same as averaging over the offset $\theta$. This again reflects the fact that according to \eqref{eq:thetashift} the spectrum of $\We(\theta)$ and $\We(\theta=0)$ are connected by a shift of the unit circle.

\subsection{The main results}

Our main result in this paper is that Anderson localization occurs naturally for electric quantum walks. Indeed, we can show that this is the case for almost all choices of the electric field, exceptions being rational fields which are known to lead to absolutely continuous spectrum and irrational fields, which have exceptionally good approximations by continued fraction expansions \cite{ewalks}. Moreover, we show that for all irrational fields the absolutely continuous spectrum is empty:
\begin{theorem}\label{thm:loclEQW}
  Let $\We$ be the electric quantum walk defined in \eqref{eq:EwalkII} on $\ell_2(\integers)\otimes\complex^2$ with coin $C_0$ satisfying $\abs{a}<1$. Then:
  \begin{enumerate}
    \item[(i)] for all irrational fields $\EF\in[0,2\pi]$ and all offsets $\theta\in[0,2\pi]$ the absolutely continuous spectrum of $\We$ is empty, i.e. $\sigma_{ac}(\We)=\emptyset$,
    \item[(ii)] for almost all electric fields $\EF\in [0,2\pi]$ and all offsets $\theta\in[0,2\pi]$ $\We$ exhibits Anderson localization, i.e. $\sigma(\We)=\sigma_{pp}(\We)$ and all eigenfunctions decay exponentially.
  \end{enumerate}
\end{theorem}
Here, up to a set of zero Lebesgue measure, the set of electric fields for which we show Anderson localization corresponds to Diophantine fields, i.e. fields $\EF$ for which there is some $A>0$ such that
\begin{equation*}
  \norm{k \EF/(2\pi)}>c\abs k^{-A}
\end{equation*}
for all $k\in\integers\backslash\{0\}$ and $c>0$. In the following we denote the set of Diophantine fields with $\DC$.

\begin{remark}~ %
  \begin{enumerate}
    \item   By \eqref{eq:thetashift} it is enough to prove Theorem \ref{thm:loclEQW} for the special case $\theta=0$.
    \item   In the following, we will concentrate exclusively on coins $C_0$ satisfying $\abs{a}>0$, i.e. coins that are not completely off-diagonal. Our proof techniques based on transfer matrices are not directly applicable in that off-diagonal case. However, since the choice $a=0$ corresponds to a flip of the internal degree of freedom of the walker and therefore effectively implements reflective boundary conditions at every lattice site, $\We$ is block-diagonal and Anderson localization is immediate. Indeed, in this special case even dynamical localization can be shown very easily by observing that the particle is trapped at its original position for all times. For a more general argument which is valid also in the random case, see \cite[Lemma 4.8]{SpaceRandom}.
  \end{enumerate}
\end{remark}

\begin{figure}[t]
\includegraphics[width=0.99\textwidth]{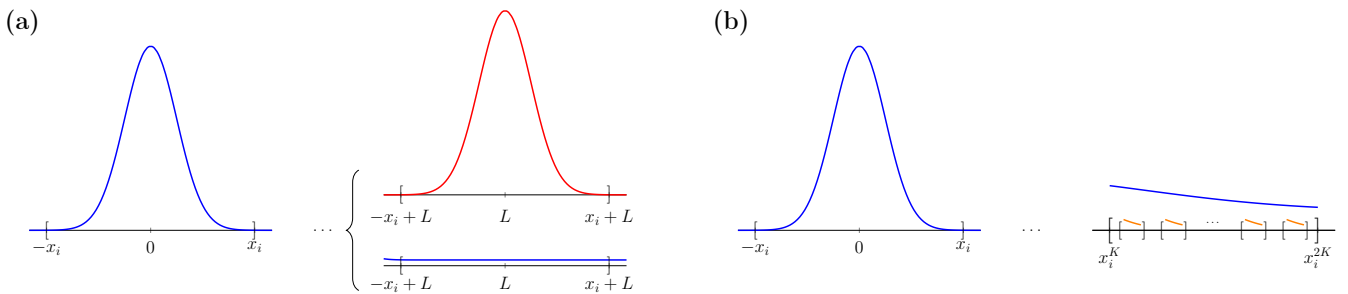}
\caption{\label{fig:proof_schema}Proof strategy: Decay estimates for eigenfunctions. (a) Given a sufficiently dense sequence of length scales $(x_i)_i$ for which we can ensure exponential decay of the eigenfunctions around origin, we have to ensure that this eigenfunction does not have another resonance at some distant $L$ (red vs. blue curve). This is achieved by excluding a small set of fields (going to zero with $x_i$ going to infinity).
(b) The so-called paving property then allows to lift the exponential decay of the eigenfunction on interval of length $x_i$ included in the interval $[x_i^\capconst,x_i^{2\capconst}]$ for some $\capconst>0$ to this larger interval (with slightly worse decay rate).}
\end{figure}
Let us quickly describe the steps used in the proof of the second statement of Theorem \ref{thm:loclEQW} displayed in Figure \ref{fig:proof_schema}. In general, we follow a method established by Bourgain and Goldstein in \cite{bourgain_goldstein}. The basic idea is to take a polynomially bounded generalized eigenfunction whose existence for all spectral points $z\in\sigma(\We)$ is guaranteed by the transfer matrices and show that it decays exponentially. The main steps in the proof are the following:
\begin{enumerate}
  \item fix some generalized eigenvalue $z\in\torus$, whose corresponding eigenfunction is polynomially bounded by Lemma \ref{lem:gen_eig},
  \item show that for a sufficiently dense sequence $x_i\to\infty$ the corresponding resolvent of the finite restriction of $\We$ to $[-x_i,x_i]$ grows exponentially with the size of the interval,
  \item for each $x_i$, show that up to a small set of fields $\EF$ with measure going to zero eventually, all finite resolvents on intervals of length $x_i$ contained in $[x_i^\capconst,x_i^{2\capconst}]$ for some $\capconst>0$ are exponentially decaying,
  \item it follows from the resolvent identity that the finite resolvent on $[x_i^\capconst,x_i^{2\capconst}]$ is exponentially decaying, which finishes the proof since we can cover all of $\integers$ with such intervals.
\end{enumerate}
The set of fields which has to be excluded to avoid the double resonances between the initial interval and the intervals in $[x_i^\capconst,x_i^{2\capconst}]$ is of measure zero and corresponds to the Diophantine fields which have to be excluded in the theorem above.

\subsection{Implications for related models}

\subsubsection{Electric quantum walks as CMV matrices}

In the literature there are several models related to the shift-coin quantum walks defined in \eqref{eq:defShCoin}. Interesting from a mathematical perspective are CMV matrices which are unitary five-diagonal matrices related to the orthogonal polynomials on the unit circle in the same sense as Jacobi matrices are related to orthogonal polynomials on the real line \cite{CMV,SimonOPUC_books}. Their doubly infinite, so-called ``extended'' variant is defined by an infinite sequence of unitary $2\times2$ building blocks
\begin{equation}\label{eq:CMV_building_block}
  \Theta_x=\begin{pmatrix}\alpha_x^*    &   \rho_x  \\  \rho_x  &   -\alpha_x\end{pmatrix},\qquad\rho_x=\sqrt{1-\abs{\alpha_x}^2},
\end{equation}
where the $\alpha_x\in\mathbb D$ are called \emph{Verblunsky coefficients}. These building blocks define a general CMV by
\begin{equation}\label{eq:CMV}
  \mathcal E=\Big(\bigoplus_{x\text{ even}}\Theta_x\Big)\Big(\bigoplus_{x\text{ odd}}\Theta_x\Big).
\end{equation}

Recently, Anderson localization was proved for CMV matrices with analytic quasi-periodic Verblunsky coefficients in \cite{wang_damanik} employing the methods from \cite{bourgain_goldstein}. Using the unitary equivalence between CMV matrices and quantum walks established in \cite{cantero2010matrix,CGMV}, it was claimed in \cite{wang_damanik} that this result carries over directly to \qws with analytic quasi-periodic coins. Yet, the proof in \cite{wang_damanik} is not directly applicable to electric quantum walks because the quasi-periodicity of the coins does not necessarily imply the quasi-periodicity of the corresponding CMV matrix:
\begin{lemma}
  The Verblunsky coefficients of the electric walk \eqref{eq:EQW} are given by
  \begin{equation*}
    \alpha_{2x}=-ce^{-i(x^2\EF+x(\arg a+\arg d)+2x\theta+\arg a)},\qquad\alpha_{2x+1}=0,
  \end{equation*}
  where $a,b,c,d$ are the entries of the constant coin $C_0$.
\end{lemma}
\begin{proof}
  To recover the CMV matrix corresponding to a walk of the form \eqref{eq:defShCoin} one needs to find a base change which makes the $a_x$ and the $d_x$ real \cite{cantero2010matrix,CGMV}. For electric walks \eqref{eq:EQW}, this base change is implemented by the diagonal matrix $\Lambda$ defined by
  \begin{align*}
    \Lambda\delta_{2x}  & =e^{-i(\frac x2(x-1)\EF+x(\arg a+\theta))}\delta_{2x}, \\
    \Lambda\delta_{2x+1}    &=e^{i(\frac x2(x+1)\EF+x(\arg a+\theta))}\delta_{2x+1}.
  \end{align*}
  The Verblunsky coefficients can be read off from $\mathcal E=\Lambda^*\We\Lambda$.
\end{proof}
By this lemma, the Verblunsky coefficients of electric walks are not quasi-periodic, i.e. their $x$-dependence is not determined by the shift $\tau_\EF(\theta)=\EF+\theta$ but rather by the skew-shift $(x,y) \mapsto (x + y, y + 2\EF)$ applied to the initial vector $(x_0,y_0)=(\arg a, \arg a + \arg d + \EF + 2\theta )$. We emphasize that the methods of \cite{bourgain_goldstein} and in particular the derivation of the large deviation estimate do not directly apply for skew-shift models. One reason for this is that the growth rate of the norm of products of the transfer matrix might not be uniformly bounded in the length of the product which requires the use of different methods like e.g. the avalanche principle \cite{goldstein_schlag,bgs_skew}. While we believe that the methods established for skew-shift models in Hamiltionan systems (see, e.g., \cite{bgs_skew,kruger2012skew,tao_skew,bourgain_book}) possibly carry over to skew-shift CMV matrices, we here use a more direct approach via the transfer matrices \eqref{eq:transMatEF} to prove localization for electric quantum walks.

\subsubsection{Anderson localization for skew-shift CMV matrices}

Yet, Theorem \ref{thm:loclEQW} implies localization for ``full'' CMV matrices with skew-shift Verblunsky coefficients
\begin{equation}\label{eq:electric_verblunskys}
  \tilde\alpha_{x}=\lambda e^{-i(x^2\EF+x(\theta+\xi)+\zeta)},\qquad \abs\lambda<1,\;x\in\integers,
\end{equation}
without the restriction that every second Verblunsky coefficient vanishes. This supplements the result in \cite{krueger_skew_shift} where pure point spectrum for half-line CMV matrices is shown for almost all $\zeta$ and $\theta=0=\xi$. In \eqref{eq:electric_verblunskys} $\xi$ is a fixed phase factor while, as above, $\theta$ denotes the random parameter.
This follows directly from a technique known as ``sieving'':
\begin{lemma}\label{lem:sieving}
  Let $\mathcal E$ be a CMV matrix whose even/odd Verblunsky coefficients vanish, e.g. a CMV matrix derived from a shift-coin quantum walk. Then
  \begin{equation*}
    \mathcal E^2=\tilde{\mathcal E}\oplus\tilde{\mathcal E}^{\,T},
  \end{equation*}
  where $\tilde{\mathcal E}$ is a CMV matrix with Verblunsky coefficients $\tilde\alpha_x=\alpha_{2x+1}$.
\end{lemma}
This statement can be verified by hand. For a detailed derivation we refer the reader to \cite{bucaj2019localization}. Combining Lemma \ref{lem:sieving} with Theorem \ref{thm:loclEQW} we infer:
\begin{corollary}
  Let $\mathcal E$ be a CMV matrix with Verblunsky coefficients given in \eqref{eq:electric_verblunskys}. Then for almost all $\EF$ and all $\theta$ the spectrum of $\mathcal E$ is pure point with exponentially decaying eigenfunctions.
\end{corollary}

\subsubsection{Anderson localization for quasi-periodic unitary band matrices}

Another model intimately related to quantum walks is the unitary band matrices introduced in \cite{AlainUnitaryBandMats} and afterwards studied e.g. in \cite{Joye2004,Hamza2009,deOliveira2007}. They are a generalization of CMV matrices in the sense that they have the form \eqref{eq:CMV} but its building blocks are arbitrary unitary $2\times2$ matrices instead of having the symmetric form in \eqref{eq:CMV_building_block}. Therefore, unitary band matrices have additional phases representing the determinants of the building blocks. Similarly to the lemma above, for these unitary band matrices one can show the following:
\begin{lemma}
  Let $W$ be a shift-coin quantum walk as in \eqref{eq:defShCoin} with coin $C=\bigoplus_xC_x$. Then
  \begin{equation*}
    W^2=U\oplus \tilde U,
  \end{equation*}
  where $U=(\bigoplus_{x\textnormal{ even}}C_x)(\bigoplus_{x\textnormal{ odd}}C_x)$ is a unitary band matrix and $\tilde U=(\bigoplus_{x\textnormal{ odd}}C_x)U(\bigoplus_{x\textnormal{ odd}}C_x)^*$.
\end{lemma}
This lemma provides a direct correspondence between the squares of shift-coin quantum walks and unitary band matrices. In one of the models studied in \cite{AlainUnitaryBandMats} the determinants of the building blocks are given by quasi-periodic functions. For Liouville frequencies, the authors prove purely singular continuous spectrum for almost all offsets. Using the above correspondence between unitary band matrices and shift-coin quantum walks, Theorem \ref{thm:loclEQW} completes this picture:

\begin{corollary}
  Let $U=(\bigoplus_{x\textnormal{ even}}C_x)(\bigoplus_{x\textnormal{ odd}}C_x)$ with building blocks $C_x=C_0e^{i(\EF x+\theta)}$. Then for almost all frequencies $\EF$ and all offsets $\theta$, $U$ exhibits Anderson localization.
\end{corollary}

\section{The absence of absolutely continuous spectrum}
\label{sec:no_ac}
In this section we prove the first part of Theorem \ref{thm:loclEQW}, i.e. we show that electric walks with irrational fields cannot have any absolutely continuous spectrum. To this end, let $(\Omega,\mathcal A,\mu,\tau)$ be an ergodic dynamical system and let $A:\Omega\to GL(d,\complex)$ be a random variable inducing an ergodic integrable cocycle $A_n(\omega):=A(\tau^n\omega)A(\tau^{n-1}\omega)\cdots A(\omega)$ with Lyapunov exponent
\begin{equation}\label{eq:lyap_exp}
  \lyap(A)=\lim_{n\to\infty}\lyap_n(A):=\lim_{n\to\infty}\frac1n\int_\Omega\log\norm{A_n(\omega)}\:\mu(d\omega).
\end{equation}
This limit exists by Kingman's subadditive theorem and is equivalent to $\lim_{n\to\infty}\log\norm{A_n(\omega)}/n$ for almost every $\omega$ \cite{cyconschroeder}. We are almost exclusively interested in quasi-periodic cocycles for which $\Omega=\torus$ and the ergodic map is given by the shift $\tau_\EF(\theta)=\EF+\theta$.

The random variables $A$ we consider correspond to transfer matrices of unitary band matrices and, in particular, of quantum walks and depend additional on the generalized eigenvalue $z\in\complex$. For the corresponding cocycles a variant of the Ishii-Pastur Theorem for unitary band operators allows to characterize the absolutely continuous part of their spectrum by the Lyapunov exponent $\lyap(A(z))\equiv\lyap(z)$ \cite{cyconschroeder,AlainUnitaryBandMats,SimonOPUC_books}. Adapted to quantum walks of the form \eqref{eq:defShCoin} with random coins we have:
\begin{theorem}\label{thm:ishiipastur}
  Let $W_\omega$ be a shift-coin walk with ergodic coin $\omega\mapsto C_\omega$. Then, for $\mu$-almost every $\omega$,
  \begin{equation*}
    \sigma_{ac}(W_\omega)\subseteq\{z\in\torus:\lyap(z)=0\}.
  \end{equation*}
\end{theorem}

Thus, the first statement in Theorem \ref{thm:loclEQW} is a direct consequence of the following result:

\begin{proposition}\label{thm:poslyap}
  Let $W_\omega$ be a shift-coin \qw \eqref{eq:defShCoin} with the coin defined as a random function on an ergodic dynamical system $(\Omega,\mathcal A,\mu,\tau)$. Then, denoting by $dz$ the Lebesgue measure on $\torus$, the Lyapunov exponent $\lyap(z)$ of the corresponding cocycle satisfies
  \begin{equation*}
    \int_\torus\lyap(z)\:dz=\int_\Omega\log\frac1{\abs{a_\omega}}\:d\mu(\omega).
  \end{equation*}
  In particular, for electric walks $\We$ with $\EF/(2\pi)$ irrational we have
  \begin{equation*}\label{eq:pos_lyap}
    \lyap=\log\frac1{\abs a}.
  \end{equation*}
\end{proposition}

We prove Proposition \ref{thm:poslyap} by tracing back our case to that of $SL(2,\reals)$-valued cocycles via the following lemma \cite{SimonOPUC_books}:
\begin{lemma}\label{lem:SU11}
  Let $\mathbb{SU}(1,1)$ be the group of unimodular unitary $2\times2$ matrices satisfying
  \begin{equation*}
    A^*\pauli_3A=\pauli_3,
  \end{equation*}
  where $\pauli_3$ denotes the third Pauli matrix. Then
  \begin{equation*}
    Q^*\mathbb{SU}(1,1)Q=SL(2,\reals),
  \end{equation*}
  where $Q=-(1+i)^\inv\begin{pmatrix}1&-i\\1&i\end{pmatrix}$.
\end{lemma}
The proof of this lemma is straightforward.
The second ingredient to the proof of Theorem \ref{thm:poslyap} is the following result known as Herman-Avila-Bochi formula \cite{herman,AvilaBochi}:
\begin{lemma}
  Let $(\Omega,\mathcal A,\mu,\tau)$ be an ergodic dynamical system and $A:\Omega\to SL(2,\reals)$ a measurable function such that the induced cocycle is $\mu$-integrable. Then, writing $R_2(\theta)=\exp[i\theta\pauli_2]$, the Lyapunov exponent $\lyap(AR_2(\theta))$ of the cocycle induced by $(AR_2(\theta))(\omega):=A(\omega)R_2(\theta)$ satisfies
  \begin{equation}\label{eq:herman_avila_bochi}
    \int_\torus\lyap(AR_2(\theta))\:d\theta=\int_\Omega\log\left(\frac{\norm{A(\omega)}+\norm{A(\omega)}^\inv}{2}\right)\:d\mu(\omega).
  \end{equation}
\end{lemma}

\begin{proof}[Proof of Theorem \ref{thm:poslyap}]
Let $\trans_x(\omega,z)=\trans(\tau^x\omega,z)$ be the transfer matrices of $W_\omega$ where $\trans(\cdot,z):\Omega\to GL(2,\complex)$ is the random variable
\begin{equation*}
  \trans(\omega,z)=\frac1{a_{\omega_0}}\begin{pmatrix} \frac{\det C_{\omega_0}}z&c_{\omega_0}\\-b_{\omega_0}&z\end{pmatrix}.
\end{equation*}
Denote by $\widetilde\trans(\omega,z)=\det(\trans(\omega,z))^{-1/2}\trans(\omega,z)$ the corresponding unimodular random variable where according to \eqref{eq:transfMat} $\det(\trans(\omega,z)) = (\det C_{\omega_0}+b_{\omega_0}c_{\omega_0})a_{\omega_0}^{-2}\in\torus$. Clearly, for $z\in\torus$, $\widetilde\trans(\omega,z)\in\mathbb{SU}(1,1)$. In particular, writing $R_3(\arg z)=\diag(z,z^\inv)$ we have
\begin{equation*}
  R_3(\arg z/2)\widetilde\trans(\omega,z)R_3(\arg z/2)=\frac1{(a_\omega d_\omega)^{1/2}}\begin{pmatrix}1&c_\omega\\-b_\omega&1\end{pmatrix}=\widetilde\trans(\omega,z=1),
\end{equation*}
such that by
\begin{equation*}
  Q^*R_3(\arg z)Q=R_2(\arg z),
\end{equation*}
we find
\begin{equation*}
  Q^*\widetilde\trans(\omega,z)Q  =R_2(\arg z/2)A(\omega)R_2(\arg z/2)
\end{equation*}
with $A:\Omega\to SL(2,\reals)$. By unitary invariance of the norm, the Lyapunov exponents of the cocycles $AR_2(\arg z)$ and $R_2(\arg z/2)AR_2(\arg z/2)$ agree which implies
\begin{equation*}
  \lyap(T)=\lyap(\widetilde T)=\lyap(R_2(\arg z/2)AR_2(\arg z/2))=\lyap(AR_2(\arg z)).
\end{equation*}
Plugging this into the Herman-Avila-Bochi formula \eqref{eq:herman_avila_bochi} we obtain
\begin{align*}
  \int_\torus\lyap(T(\omega,z))\:d(\arg z)&=\int_\Omega\log\left(\frac{\norm{A(\omega)}+\norm{A(\omega)}^\inv}{2}\right)\:d\mu(\omega).
\end{align*}
Moreover, since $A\in SL(2,\reals)$, the right-hand side can be further evaluated to
\begin{equation*}
  \norm{A(\omega)}+\norm{A(\omega)}^\inv=\sqrt{\tr(A(\omega)^*A(\omega))+2}=\sqrt{\tr(T(\omega)^*T(\omega))+2},
\end{equation*}
where $T(\omega)=T(\omega,z=1)$ and we used the cyclic invariance of the trace. The trace is easily calculated and gives
\begin{equation*}
  \tr(T(\omega)^*T(\omega))+2=\frac4{\abs{a_\omega}^2},
\end{equation*}
which proves the first part of the theorem.

For the special case of electric walks with transfer matrices defined in \eqref{eq:transMatEF}, integrating over $z$ is equivalent to taking the expectation value \wrt the offset $\theta$. It follows from $\abs{a_\omega}=\abs a$ that
\begin{equation*}
  \lyap(T)=\int_\torus\lyap(T(\theta,z))\:d(\arg z)=\log\frac1{\abs a}.
\end{equation*}
\end{proof}
For $0<\abs{a}<1$ by Theorem \ref{thm:ishiipastur} the absolutely continuous part of the spectrum is empty for almost all offsets $\theta$. Since the offset merely shifts the spectrum by the definition of the model in \eqref{eq:EQW}, this result holds for all $\theta$. This proves the first part of Theorem \ref{thm:loclEQW}.

\section{Proof of Anderson localization}\label{sec:proof_loc}
In the following we present the steps of the proof of Theorem \ref{thm:loclEQW} (ii) in detail. In the first section we discuss products of transfer matrices of electric walks and provide a large deviation estimate. This allows us to derive estimates on the decay of the resolvents of finite restrictions of electric walks in the next section, which in turn implies the decay of generalized eigenfunctions. Carrying over a result by Bourgain and Goldstein which excludes double resonances to the walk setting we assemble the proof in Section \ref{sec:assemble_loc_proof}.

\subsection{Products of transfer matrices}
In this section we provide some results about the behaviour of products of transfer matrices for electric quantum walks $W=\We$ defined in \eqref{eq:transMatEF}. We begin with a large deviation estimate for the growth rate of their norm. To this end, we extend the offset $\theta$ or, equivalently, the spectral parameter $z$ from the unit circle to the complex plane. We interpret the coin entry $a_x=ae^{i(\EF x+\theta)}$ of $\We$ as a function
\begin{equation*}
  a:\torus\to\mathbb D,\quad \theta\mapsto a(\theta)=ae^{i\theta},
\end{equation*}
such that $a_x=a(\EF x+\theta)$. This function is analytic in $\theta$ and can thus be extended to a bounded analytic function on the strip $\{\theta+i\lambda:\abs\lambda<\rho\}$ for $\rho>0$. For fixed $\rho>0$ this extension is bounded by $\sup_{\abs\lambda<\rho}\abs{a(\theta+i\lambda)}=ae^{\rho}$ and we define the constant
\begin{equation}\label{eq:C_a}
  \capconst_a:=\log\frac2{\sup_{\abs\lambda<\rho}\abs{a(\theta+i\lambda)}}=\log2 -\log a-\log\rho.
\end{equation}
For reasons which become clear below we choose the strip such that $\capconst_a>0$. Denoting by $\dr_\kk$ the denominators of the continued fraction approximants of $\EF/(2\pi)$ and defining
\begin{equation*}
  \beta(\EF):=\limsup_{\kk\to\infty}\frac{\log\dr_{\kk+1}}{\dr_\kk}
\end{equation*}
we have
\begin{proposition}[Large deviation estimate]\label{prop:LDE}
  Let $\kappa>0$. Then there exists $n_0(\capconst_a,\kappa)$ and constants $\smallconst_0,\smallconst_1$ independent of $\capconst_a,\kappa$ such that for all fields $\EF$ with $\beta(\EF)<\smallconst_0\kappa/\capconst_a$, $z\in\torus$ and $n>n_0(\capconst_a,\kappa)$
  \begin{equation*}
    \textrm{mes}\{\theta\in\torus:\abs{\frac1n\log\Norm{\trans_n(\theta,z)\cdots\trans_1(\theta,z)}-\lyap_n}>\kappa\}<e^{-(\smallconst_1/\capconst_a^3)\kappa^3n},
  \end{equation*}
  where $\lyap_n$ is the finite Lyapunov exponent from \eqref{eq:lyap_exp}.
\end{proposition}

\begin{proof}
  Like in the proof of Theorem \ref{thm:poslyap} the unimodular transfer matrices $\widetilde\trans_x(\theta,z):=\det(\trans_x(\theta,z))^{-1/2}\trans_x(\theta,z)$ are unitarily equivalent to the $SL(2,\reals)$-valued cocycle
  \begin{equation}\label{eq:SL2R_cocyc}
    A_x(\theta,z):=Q^*\widetilde\trans_x(\theta,z)Q
  \end{equation}
  by Lemma \ref{lem:SU11}. Note that $\norm{\widetilde\trans_x(\theta,z)}=\norm{A_x(\theta,z)}$ and the Lyapunov exponents of the induced cocycles agree by unitary invariance of the norm. Also, from
   \begin{align*}
     \norm{A_x(\theta,z)}+\norm{A_x(\theta,z)}^\inv   =\sqrt{\tr(A_x(\theta,z)^*A_x(\theta,z))+2}
                                        =\sqrt{\tr(\trans_x(\theta,z)^*\trans_x(\theta,z))+2}
                                        =\frac2{\abs{a_x}},
   \end{align*}
   we infer that for the analytic extension of $A_x(\theta,z)$ we have $\sup_{\abs\lambda<\rho}\norm{A_x(\theta+i\lambda,z)},\sup_{\abs\lambda<\rho}\norm{A_x(\theta+i\lambda,z)^\inv}\leq\frac2{\sup_{\abs\lambda<\rho}\abs{a(\theta+i\lambda)}}$ and therefore
   \begin{equation*}
     \sup_{\abs\lambda<\rho}\frac1n\log\norm{A_n(\theta+i\lambda,z)\cdots A_1(\theta+i\lambda,z)}\leq \capconst_a.
   \end{equation*}
   The proposition therefore follows from \cite[Theorem 1]{you_zhang}.
\end{proof}

\begin{remark}~
  \begin{enumerate}
    \item For Diophantine fields $\beta(\EF)=0$ such that the condition $\beta(\EF)<\smallconst_0\kappa/\capconst_a$ is trivially fulfilled for $\capconst_a>0$. In contrast, for $\EF/(2\pi)$ Liouville one has $\beta(\EF)>0$.
    \item Since the transfer matrices \eqref{eq:transMatEF} only depend on the combined parameter $z^\inv e^{i\theta}$, this lemma can be interpreted as a result on the spectral parameter for fixed $\theta$.
    \item Choosing the strip for the analytic extension of $a(\theta+i\lambda)$ too large results in $\capconst_a$ negative. In this case the large deviation estimate becomes trivial. In contrast, in \cite{you_zhang,wang_damanik} $\capconst_a>0$  whenever it is well-defined, see the next remark.
    \item In \cite{wang_damanik} one has to assume $\sup_{\abs\lambda<\rho}\abs{\alpha(\theta+i\lambda)}<1$ where $\alpha$ is the quasi-periodic Verblunsky coefficient. Otherwise the constant corresponding to $\capconst_a$ is ill-defined.
    \item In contrast to the models e.g. in \cite{bourgain_goldstein,wang_damanik,you_zhang}, in our system the finite Lyapunov exponent $\lyap_n$ is independent of the spectral parameter $z$ due to the expectation value in $\theta$.
    \item This large deviation estimate was first derived in \cite{you_zhang} for quasi-periodic Schrödinger cocycles and is stronger than the \emph{sharp} large deviation estimate in \cite{goldstein_schlag} which gives a bound $<\exp[-\smallconst\kappa n]$ but applies ``only'' to  so-called \emph{strong} Diophantine fields. For usual Diophantine fields the estimate obtained in \cite{bourgain_goldstein} is weaker ($\exp[-\smallconst n^\sigma]$ for some $\sigma\in(0,1)$).
  \end{enumerate}
\end{remark}

As discussed above, we are interested in the decay properties of the resolvent of $\We$ with the decay determined by the Lyapunov exponent. However, generalized eigenfunctions are determined by repeatedly applying the transfer matrix. The following lemmas connect the two approaches. The first important result is the following \cite[Lemma 2.1]{bourgain_goldstein}:
\begin{lemma}\label{lem:lyap_upper_bound}
  Assume $\EF\in\textrm{DC}$. Then for all $z\in\torus$ and for all $\theta\in\torus$ we have
  \begin{equation}\label{eq:normprod_lyap}
    \frac1n\log\Norm{\trans_n(\theta,z)\cdots\trans_1(\theta,z)}<\lyap_n+\kappa \capconst
  \end{equation}
  for any $\kappa>0$ and $n$ large enough.
\end{lemma}
\begin{proof}
  The related $SL(2,\reals)$-cocycle $A(\theta,z)$ in \eqref{eq:SL2R_cocyc} extends to an analytic function on some annulus around the unit circle satisfying $\norm{A_n\cdots A_1(\theta+i\lambda)}+\norm{(A_n\cdots A_1(\theta+i\lambda))^\inv}<\capconst^n$. The lemma follows from the proof given in \cite{bourgain_goldstein}.
\end{proof}

The following lemma expresses the Lyapunov exponent as the average over long shift orbits. It allows us to use techniques from semi-algebraic set theory and plays a central role in the proof of Anderson localization.
\begin{lemma}\label{lem:large_averages}
  Let $\EF/(2\pi)$ satisfy the finite Diophantine condition
  \begin{equation}\label{eq:EF_n2A}
    \norm{k\EF/(2\pi)}>c\abs k^{-A},\qquad0<\abs k <n^{2A}
  \end{equation}
  for $n\in\integers$. Then, for
  \begin{equation*}
    J>n^{2A}
  \end{equation*}
  we have for all $z\in\torus$ and $\theta\in\torus$
  \begin{equation*}
    \frac1J\sum_{j=1}^J\frac1n\log\norm{\trans_n(\theta+j\EF,z)\cdots\trans_1(\theta+j\EF,z)}=\lyap_n+\mathcal O(n^\inv).
  \end{equation*}
\end{lemma}
\noindent For the sake of brevity we omit the proof which follows directly from applying that of \cite[Lemma 3.1]{bourgain_goldstein} to \eqref{eq:SL2R_cocyc}.

\subsection{Finite \un restrictions and resolvent estimates}
In the proof of localization we need estimates on the decay of the resolvent of finite restrictions of the electric walk \eqref{eq:EQW}. While restrictions of Schrödinger operators to finite intervals are easily constructed by projecting onto the respective interval, for unitary operators one has to be more careful since projecting usually destroys unitarity. In this paragraph we discuss the construction of \emph{unitary} restrictions of electric walks and establish the exponential decay of the matrix elements of their resolvents.

Let us consider a shift-coin walk of the form \eqref{eq:defShCoin}. Replacing the coin at lattice site $a\in\integers$ by
\begin{equation*}
  C_a\mapsto\alpha\sigma_1,\qquad\alpha\in\torus,
\end{equation*}
decouples the walk at $2a$, i.e. the resulting walk $W_\alpha$ satisfies
\begin{equation*}
  \Pgg {2a}W_\alpha\Pgg {2a}=W_\alpha,
\end{equation*}
and similarly for $\Pll{2a}$, where $\Pgg{2a}$ and $\Pll{2a}$ denote the projections onto $\ell_2(\{2a+1,2a+2,\dots\})$ and $\ell_2(\{\dots,2a-1,2a\})$, respectively.
This leads to the following definition:
\begin{definition}\label{def:finite_un_restrictions}
  Let $W_{\alpha}$ be the walk \eqref{eq:defShCoin} with the coin at lattice site $a\in\integers$ replaced by the reflective boundary condition $\alpha\sigma_1,\alpha\in\torus$. Then, denoting by $\chi_I$ the restriction operator to the interval $I\subset\integers$ the \emph{half-space walk} $W_\alpha^{[a,\infty)}$ is the operator on $\ell_2(\{2a+1,\dots\})$ defined by
  \begin{equation*}
    W_\alpha^{[a,\infty)}:=\chi_{[2a+1,\infty)}W_\alpha\chi_{[2a+1,\infty)}.
  \end{equation*}
  Similarly, on $\ell_2(\{\dots,2a-1,2a\})$ we define $W_\alpha^{(-\infty,a]}:=\chi_{(-\infty,2a]}W_\alpha\chi_{(-\infty,2a]}$. Replacing the coins $C_a\mapsto\alpha\sigma_1,C_b\mapsto\beta\sigma_1$ at lattice sites $a<b$ the resulting walk $W_{\alpha,\beta}$ defines the \emph{finite \un restriction} $W_{\alpha,\beta}^{[a,b]}$ on $\ell_2([2a+1,\dots,2b])$ by
  \begin{equation*}
    W_{\alpha,\beta}^{[a,b]}:=\chi_{[2a+1,2b]}W_{\alpha,\beta}\chi_{[2a+1,2b]}.
  \end{equation*}
\end{definition}

\begin{remark}
  Expressed as CMV matrices, the odd Verblunsky coefficients of walks of the form \eqref{eq:defShCoin} manifestly vanish \cite{CGMV}. As a consequence, we cannot define restrictions by setting any Verblunsky coefficients to the unit circle as in \cite{krueger_skew_shift} and therefore have to be more careful about indices. In particular, by the above method we cannot cut ``between'' cells but only within them, see Figure \ref{fig:restrictions}.
\end{remark}

\pgfdeclarelayer{background}
\pgfdeclarelayer{foreground}

\pgfsetlayers{background,main,foreground}
\begin{figure}[ht]
\begin{equation}\label{eq:res_walk_mat}
\begin{tikzpicture}[x=4cm,y=1cm,baseline=(current  bounding  box.center)]
\matrix (A) [matrix of math nodes,
             left delimiter  = (,
             right delimiter = ),
             ] at (0,0)
{
\ddots \\
\w\alpha & 0 &  0 & 0  \\
 0 & \w0 &  0 & \w\alpha\\
&0 & b_{a+1} & 0 & \w0 & a_{a+1}  \\
&   0 & d_{a+1} &0 & 0 & c_{a+1} \\
&  &     &   &   &  \ddots & & & & & \\
&& &   &   & b_{b-1} & \w0 &  \w0 &a_{b-1}& \\
& & &   &   & d_{b-1} & 0 &  0&c_{b-1}&\\
&& & & && &  \beta &0 &  \w0 &\w0 \\
& && & & &   &  0 & 0 & 0   &\beta\\
& && & & &   &   &  &    &\ddots\\
};
\begin{pgfonlayer}{background}

\draw[fill=gray!20] ($(A-2-2.north west)+(-\www,0)$) rectangle ($(A-3-3.south east)+(\www,0)$);
\draw[fill=gray!20] ($(A-3-3.south east)+(\www,0)$) rectangle ($(A-5-5.south east)+(\www,0)$);
\draw[fill=gray!20] ($(A-7-7.north west)+(0,0)$) rectangle ($(A-8-8.south east)+(\www,0)$);
\draw[fill=gray!20] ($(A-8-8.south east)+(\www,0)$) rectangle ($(A-10-10.south east)+(\www,0)$);
\draw[very thick,red] ($(A-3-3.north west)+(-\www,0)$) rectangle ($(A-9-9.south east)+(\www,0)$);
\end{pgfonlayer}
\end{tikzpicture}
\end{equation}
\caption{\label{fig:restrictions}The walk $W=CS$ as in \eqref{eq:defShCoin} with coins at $x=a,b$ replaced by reflecting coins $\alpha\sigma_1,\beta\sigma_1$, respectively. The red frame delimits the finite \un restriction $W_{\alpha,\beta}^{[a,b]}$ in Definition \ref{def:finite_un_restrictions}.}
\end{figure}

Let us fix boundary phases $\alpha,\beta\in\torus$ and denote by $R^{[a,b]}_z=(W^{[a,b]}_{\alpha,\beta}-z)^\inv$ the resolvent of $W^{[a,b]}_{\alpha,\beta}$. Its matrix elements can be written as \cite[(5.38)]{Albert}
\begin{equation}\label{eq:res_res}
  \abs{R^{[a,b]}_z(2x-i,2y-j)}=\absbig{\frac{\phi^+_{2y-1+j}\phi^-_{2x-i}}{\phi^+_{2k}\phi^-_{2k-1}-\phi^+_{2k-1}\phi^-_{2k}}}
\end{equation}
for $a<x<y<b$, $i,j=0,1$ and $a<k\leq b$ arbitrary. Here, $\phi^-$ and $\phi^+$ are left- and right-compatible with $W^{[a,b]}_{\alpha,\beta}$, i.e. $\phi^-$ solves
\begin{equation*}
  (W_\alpha^{[a,\infty)}-z)\phi=0,
\end{equation*}
whereas $\phi^+$ is a solution to
\begin{equation*}
  (W_\beta^{(-\infty,b]}-z)\phi=0.
\end{equation*}
The denominator of \eqref{eq:res_res} may be rewritten as
\begin{align*}
  \absbig{{\phi^+_{2k}\phi^-_{2k-1}-\phi^+_{2k-1}\phi^-_{2k}}}&=\absbig{\scpbig{\begin{pmatrix}-{\phi^+_{2k-1}}^*\\{\phi^+_{2k}}^*\end{pmatrix}}{\begin{pmatrix}\phi^-_{2k-1}\\\phi^-_{2k}\end{pmatrix}}}.
\end{align*}
Choosing $k=b$, this can be written as the overlap of the right-compatible solution at $b$ and the left-compatible solution transported from $a$ to $b$ by repeatedly applying the transfer matrices, i.e.
\begin{equation}\label{eq:res_scp}
  \absbig{\scpbig{\begin{pmatrix}-{\phi^+_{2b-1}}^*\\{\phi^+_{2b}}^*\end{pmatrix}}{\trans_{b-1}(z)\cdots\trans_{a}(z)\begin{pmatrix}\phi^-_{2a-1}\\\phi^-_{2a}\end{pmatrix}}}.
\end{equation}

In the following, we are interested in finite \un restrictions of electric walks defined in \eqref{eq:EwalkII} which we denote by $W_{\EF,\alpha,\beta}^{[a,b]}$. Since the position dependence of their coins in encoded only in the electric field operator we have for all offsets $\theta\in\torus$
\begin{equation*}
    W_{\EF,\alpha,\beta}^{[a,b]}(\theta)=W_{\EF,\alpha,\beta}^{[0,b-a]}(\theta+a\EF).
\end{equation*}

Bounding the scalar products in \eqref{eq:res_scp} from below together with upper bounds for $\phi^+_{2y-1+j}$ and $\phi^-_{2x-i}$ provides estimates for the matrix elements of finite resolvents $R_{\EF,z}^{[a,b]}$ of electric walks:

\begin{proposition}\label{prop:exp_decay}
  Assume that for $\epsilon>0$ there exists $n$ sufficiently large such that
  \begin{equation}\label{eq:assumption_for_exp_decay}
    \frac1{n}\log\Norm{\trans_{n}(\theta,z)\cdots\trans_1(\theta,z)}\geq\lyap_{n}-\epsilon.
  \end{equation}
  Let $\alpha,\beta\in\torus$ and $z\in\torus\backslash\sigma(W^{[0,n+1]}_{\EF,\alpha,\beta})$. Then, there exist $\alpha_0\in\{-\alpha,\alpha\}$ and $\beta_0\in\{-\beta,\beta\}$ such that
  \begin{equation*}
    \abs{R^{[0,n+1]}_{\EF,z}(2x-i,2y-j)}\leq e^{-\abs{x-y}\lyap_n+\capconst\epsilon n}
  \end{equation*}
  for all $1\leq x,y\leq n+1$.
\end{proposition}
\begin{proof}
  Assume $x<y$. By \eqref{eq:res_res} we have to bound the scalar product \eqref{eq:res_scp} from below and $\abs{\phi^+_{2y-1+j}}$ and $\abs{\phi^-_{2x-i}}$ from above. Since $\phi^-$ is left-compatible, we read off from \eqref{eq:res_walk_mat} that it satisfies
  \begin{equation*}
    z\phi^-_{2a-1}=\alpha\phi^-_{2a}\quad\Rightarrow\quad \abs{\phi^-_{2a-1}}=\abs{\phi^-_{2a}}
  \end{equation*}
  for $a=1$. We therefore fix $(\phi^-_{2a-1},\phi^-_{2a})=(1,z\alpha^*)$ as boundary values and similarly $(\phi^+_{2b-1},\phi^+_{2b})=(1,z^*\beta)$ for $a=1$ and $b=n+1$, respectively. Then, the left-compatible solution $\phi^-$ is determined by
  \begin{equation*}
    \begin{pmatrix}\phi^-_{2x-1}\\\phi^-_{2x}\end{pmatrix}=\trans_{x-1}(\theta,z)\cdots\trans_1(\theta,z)\begin{pmatrix}1\\z\alpha^*\end{pmatrix}
  \end{equation*}
  which implies
  \begin{equation*}
    \abs{\phi^-_{2x-i}}\leq\sqrt2\Norm{\trans_{x-1}(\theta,z)\cdots\trans_1(\theta,z)}\leq e^{(x-1)\lyap_{x-1}+o(x-1   )},
  \end{equation*}
  for $i=0,1$ where in the second inequality we used \eqref{eq:normprod_lyap}. A similar estimate holds for $\phi^+_{2y-1+j}$ for which
  \begin{equation*}
    \begin{pmatrix}\phi^+_{2y-1}\\\phi^+_{2y}\end{pmatrix}=\trans_{y}(\theta,z)^\inv\cdots\trans_{n}(\theta,z)^\inv\begin{pmatrix}1\\z^*\beta\end{pmatrix}
  \end{equation*}
  such that
  \begin{equation*}
    \abs{\phi^-_{2y-1+j}}\leq\sqrt2\norm{\trans_y^\inv(\theta,z)\cdots\trans_{n}^\inv(\theta,z)}\leq e^{(n-y)\lyap_{n-y}+o(n-y)}.
  \end{equation*}

  The lower bound is achieved as follows: denote by $\phi^{\alpha,\beta}(z)$ the scalar products \eqref{eq:res_scp} for $a=1,b=n+1$. Then
  \begin{equation*}
    \begin{pmatrix} \phi^{\alpha,\beta}(z) & \phi^{\alpha,-\beta}(z)   \\ \phi^{-\alpha,\beta}(z) &  \phi^{-\alpha,-\beta}(z) \end{pmatrix}=\begin{pmatrix} 1 & z\beta^* \\ 1  & -z\beta^*  \end{pmatrix}\trans_{n}(\theta,z)\cdots\trans_1(\theta,z)\begin{pmatrix} 1   &   1  \\ z\alpha^*  & -z\alpha^*   \end{pmatrix}
  \end{equation*}
  which together with the assumption \eqref{eq:assumption_for_exp_decay} gives the desired lower bound
  \begin{equation*}
    \max_{\pm\alpha,\pm\beta}\abs{\phi^{\alpha,\beta}(z)}\geq \capconst\Norm{\trans_{n}(\theta,z)\cdots\trans_1(\theta,z)}\geq e^{n\lyap_{n}-n\epsilon}.
  \end{equation*}
  Putting everything together we obtain
  \begin{align*}
    \abs{R^{[0,n+1]}_{\EF,z}(2x-i,2y-j)}&\leq e^{(x-1)\lyap_{x-1}+o(x-1)+(n-y)\lyap_{n-y}+o(n-y)-n\lyap_{n}+n\epsilon}  \\
    &\leq e^{(x-1)\lyap_{x-1}+(n-y)\lyap_{n-y}-n\lyap_{n}+\capconst n\epsilon} \\
    &\leq e^{x\lyap_{x-1}-y\lyap_{n-y}+\capconst n\epsilon}  \\
    &\leq e^{(x-y)\lyap_n+\capconst n\epsilon}
  \end{align*}
\end{proof}

In the proof of localization we want to infer from this exponential decay of the matrix elements of the finite resolvent that of the generalized eigenvectors of the infinite problem. To this end, one has to find a way to express the the generalized eigenvector by these matrix elements. For Schr\"{o}dinger operators, restricting $(H-z)\psi=0$ to the interval $[a,b]$ yields $\psi_n=-(H^{[a,b]}-z)^\inv(n,a)\psi_{a-1}-(H^{[a,b]}-z)^\inv(n,b)\psi_{b+1}$ for $a<n<b$. For shift-coin \qws $W=CS$ there is a similar expression which, however, involves more boundary terms. To see this, define the unitary matrices $\mathcal L=C\left(\bigoplus_x\pauli_1\right)$ and $\mathcal M=\left(\bigoplus_x\pauli_1\right)S$ which are both block diagonal but with blocks that are shifted by one index relative to each other. Then $W=\mathcal L\mathcal M$, and solutions to $(W-z)\psi=0$ also solve
\begin{equation*}
  (z\mathcal L^*-\mathcal M)\psi=0.
\end{equation*}
The operator on the left hand side has a particularly simple structure:
\begin{lemma}
  The operator $A=z\mathcal L^*-\mathcal M$ is tridiagonal with diagonal entries
  \begin{equation*}
      A_{2x,2x}=zb_x,\qquad A_{2x+1,2x+1}=zc_x
  \end{equation*}
  and off-diagonal terms
  \begin{equation*}
      A_{2x,2x+1}=za_x,\:A_{2x+1,2x}=zd_x,\:A_{2x-1,2x}=A_{2x,2x-1}=-1.
  \end{equation*}
\end{lemma}
From this we obtain the following expression for the generalized eigenfunctions:
\begin{lemma}\label{lem:poisson_formula}
  Let $\psi$ solve $W\psi=z\psi$.
  Then, for $2a+1<x<2b$
  \begin{equation*}
    \psi_x =G_{z}^{[a,b]}(x,2a+1)\left(zc_a\psi_{2a+1}-\psi_{2a+2}\right)+G_{z}^{[a,b]}(x,2b)\left(zb_b\psi_{2b}-\psi_{2b-1}\right),
  \end{equation*}
  where $G_z^{[a,b]}(x,y):=\scp{\delta_x}{(z(\mathcal L^{[a,b]}_{\alpha,\beta})^*-\mathcal M^{[a,b]}_{\alpha,\beta})^\inv\delta_y}$.
\end{lemma}
\begin{proof}
  As in Definition \ref{def:finite_un_restrictions}, denote by $\chi_I$ the restriction operator to the interval $I\subset\integers$ such that $\psi^{[a,b]}=\chi_{[2a+1,2b]}\psi$. Then, since $\psi$ is a generalized eigenfunction the statement is evident from
  \begin{equation*}
    \left(z(\mathcal L^{[a,b]}_{\alpha,\beta})^*-\mathcal M^{[a,b]}_{\alpha,\beta}\right)\psi^{[a,b]}=-\chi_{[a,b]}(z\mathcal L^*-\mathcal M)(\psi-\psi^{[a,b]})
  \end{equation*}
  by inverting the operator on the left side and using the tridiagonality of $z\mathcal L^*-\mathcal M$.
\end{proof}

Clearly, from $R_z=(W-z)^\inv=-(z\mathcal L^*-\mathcal M)^\inv\mathcal L^\inv=-G_z\mathcal L^\inv$ it follows that
\begin{equation*}
  \absbig{G_{z}^{[a,b]}(x,y)}  \leq\sum_k\absbig{R_{z}^{[a,b]}(x,k)\mathcal L_{ky}}\leq\sum_k\absbig{R_{z}^{[a,b]}(x,k)}
\end{equation*}
where the sum over $k$ is finite. This lemma together with Proposition \ref{prop:exp_decay} therefore implies the exponential decay of generalized eigenfunctions of electric walks provided that \eqref{eq:assumption_for_exp_decay} holds which we establish next by excluding double resonances.

\subsection{Excluding double resonances}

The missing piece in the proof of Theorem \ref{thm:loclEQW} (ii) is the verification of \eqref{eq:assumption_for_exp_decay}. To this end, we use the following lemma which roughly speaking states that if on some interval around the origin we are close to a generalized eigenvalue of $\We$, the probability in the field $\EF$ that \eqref{eq:assumption_for_exp_decay} fails on some interval far out is exponentially small. Let
\begin{equation*}\label{eq:C(a)}
  \capconst(a)=\sup_{n\in\integers}(ne^{n\capconst_a})^{1/n}=2\times3^{1/3}\abs a^\inv e^{-\rho},
\end{equation*}
with $\capconst_a$ as in \eqref{eq:C_a}, which satisfies $\capconst(a)>0$ by $0<\abs a<1$. Moreover, the choice of $\rho$ as above which ensured $\capconst_a>0$ guarantees that $\capconst(a)>1$.

\begin{lemma}\label{lem:no_double_resonances}
  Let $\kappa>0$ and $n\in\integers$ be sufficiently large. Denote by $S_{n,\kappa}\subset\torus$ the set of fields $\EF\in\DC$ such that there are $n_0<n^\capconst$, $\ell\in[2^{(\log n)^2},2^{(\log n)^3}]$, $z\in\torus$ with
  \begin{equation}\label{eq:double_res_1}
    \norm{R^{[0,n_0]}_{\EF,z}}>\capconst(a)^n,
  \end{equation}
  and
  \begin{equation}\label{eq:double_res_2}
    \frac1n\log\norm{\trans_n(\ell\EF,z)\cdots\trans_1(\ell\EF,z)}<\lyap_n-\kappa.
  \end{equation}
  Then
  \begin{equation*}
    \mes(S_{n,\kappa})<2^{-\frac14(\log n)^2}
  \end{equation*}
\end{lemma}

The proof of this important result combines semi-algebraic set theory with measure estimates on sections of $\torus^2$ which are unions of finitely many intervals. It is an straightforward adaption of the proof given in \cite{bourgain_goldstein} to the walk setting and, for the sake of brevity, we only recall its main steps:
\begin{enumerate}
  \item a measure estimate on the set $S$ of $(\EF,\theta)\in\torus\times\torus$ such that $\EF$ satisfying \eqref{eq:EF_n2A}, and \eqref{eq:double_res_1} and \eqref{eq:double_res_2} hold with $\ell\EF$ replaced by $\theta$,
  \item reformulating \eqref{eq:double_res_1} and \eqref{eq:double_res_2} in terms of semi-algebraic sets implies a complexity bound on the the $\theta$-sections of $S$ in the sense that they are given as unions of finitely many intervals. Here one needs to extend the real energy parameter $E$ of the Hamiltonian in \cite{bourgain_goldstein} to complex $z=(\re z,\im z)$ for the unitary walk which, however, does not influence the complexity bounds.
  \item a measure estimate on sections of $\torus\times\torus$ which are unions of finitely many intervals.
\end{enumerate}

Below, we use Lemma \ref{lem:no_double_resonances} as follows: assuming $\EF\notin S_{n,\kappa}$ for any $n$ and $\kappa$ and showing that \eqref{eq:double_res_1} holds, we conclude via \eqref{eq:double_res_2} that the condition in Proposition \ref{prop:exp_decay} is fulfilled.

\subsection{Anderson localization for Diophantine fields}\label{sec:assemble_loc_proof}

\begin{proof}[Proof of Theorem \ref{thm:loclEQW} (ii)]
  For fixed $n\in\integers$ and $\kappa>0$ let us denote by $S_{n,\kappa}$ the set obtained in Lemma \ref{lem:no_double_resonances}. Then, with
  \begin{equation*}
    S_\kappa=\bigcap_{n'}\bigcup_{n>n'}S_{n,\kappa},\qquad S=\bigcup_\kappa S_\kappa
  \end{equation*}
  we have
  \begin{equation*}
    \mes(S_\kappa)\leq\inf_{n'}\sum_{n>n'}e^{-\frac14(\log n)^2}=0.
  \end{equation*}
  Then, since $S_\kappa\supset S_{\kappa'}$ for $\kappa'>\kappa$ we can pick a countable subsequence of $S_\kappa$ and conclude $\mes(S)=0$. The set $S$ constitutes the set of fields we have to remove from the Diophantine fields.

  Let us therefore assume $\EF/(2\pi)\in\DC\backslash S$ and take $z\in\torus$, $\psi\in\complex^\integers$ such that
  \begin{equation*}
    \We\psi=z\psi
  \end{equation*}
  and $\psi$ is polynomially bounded, i.e. $\abs{\psi_n}<n^\capconst$ with $\psi_0=1$. Then, according to Lemma \ref{lem:gen_eig}, $z\in\sigma(\We)$ since $\psi$ is a generalized eigenfunction.

  Moreover, since $\EF/(2\pi)$ is irrational, we know from Theorem \ref{thm:poslyap} that
  \begin{equation*}
    \lyap(\EF)\equiv\lyap=\log\abs a^\inv>0.
  \end{equation*}
  Let $\kappa\ll\lyap$. Then, since $\EF\notin S_\kappa$, there is some $n'$ such that $\EF\notin S_{n,\kappa}$ for all $n>n'$. In particular, we take $n'$ large enough such that
  \begin{equation*}
    \lyap_n<\lyap+\kappa.
  \end{equation*}

  Since $\EF/(2\pi)\notin S$, assuming that there is $n_0<n^\capconst$ such that
  \begin{equation}\label{eq:dist_to_spec}
    \norm{R_{\EF,z}^{[-n_0,n_0]}}>\capconst(a)^n
  \end{equation}
  we conclude from Lemma \ref{lem:no_double_resonances} that for all $2^{(\log n)^2}<\abs \ell<2^{(\log n)^3}$
  \begin{equation*}
    \frac1n\log\norm{\trans_n(\ell\EF,z)\cdots\trans_1(\ell\EF,z)}>\lyap_n-\kappa.
  \end{equation*}
  Thus, the assumption in Proposition \ref{prop:exp_decay} is satisfied. Properly choosing the boundary phases assures that $z$ is not an eigenvalue of any of the restrictions such that
  \begin{equation*}
    \abs{R_{\EF,z}^{[\ell,n+\ell]}(x,y)}<e^{-\abs{x-y}\lyap_n+\capconst\epsilon n}<e^{-\abs{x-y}\lyap+\capconst\epsilon n}.
  \end{equation*}
  Now, to conclude the exponential decay of the resolvent on the whole interval let $2^{(\log n)^2+1}<N<2^{(\log n)^3-1}$ and consider the interval $[N/2,2N]$. Then, it follows from the so-called paving property \cite{bourgain_goldstein,bourgain_book} that for $x,y\in[N/2,2N]$
  \begin{equation*}
    \abs{R_{\EF,z}^{[N/2,2N]}(x,y)}<e^{-\abs{x-y}\frac\lyap2+\capconst\epsilon N}.
  \end{equation*}

  It then follows from Lemma \ref{lem:poisson_formula} that for $k\in[N/2,2N]$ the generalized eigenfunction $\psi$ satisfies
  \begin{align*}
    \abs{\psi_k}   &\leq\abs{G_{\EF,z}^{[N/2,2N]}(k,N/2)}\abs{zc_{N/4}\psi_{N/2}-\psi_{N/2+1}}+\abs{G_{\EF,z}^{[N/2,2N]}(k,2N)}\abs{zb_N\psi_{2N}-\psi_{2N-1}}  \\
    &\leq\abs{G_{\EF,z}^{[N/2,2N]}(k,N/2)}N^\capconst+\abs{G_{\EF,z}^{[N/2,2N]}(k,2N)}N^\capconst
  \end{align*}
  such that for $k=N$
  \begin{equation*}
     \abs{\psi_k}\leq N^\capconst e^{-\frac\lyap4N+\capconst\epsilon N}<e^{-\frac\lyap5N},
  \end{equation*}
  which is the exponential decay of the generalized eigenfunction. By similar arguments $\psi$ decays exponentially also on the negative half axis.

  To complete the proof, we have to show that the assumption \eqref{eq:dist_to_spec} holds for some $n_0<n^\capconst$. From the proof of Lemma \ref{lem:poisson_formula} together with $\psi_0=1$ we conclude
  \begin{equation*}
    1\leq\norm{\psi^{[-n_0,n_0]}}\leq2\norm{G^{[-n_0,n_0]}_{\EF,z}}\left(\abs{\psi_{-2n_0+1}}+\abs{\psi_{-2n_0+2}}+\abs{\psi_{2n_0-1}}+\abs{\psi_{2n_0}}\right).
  \end{equation*}
  Therefore, \eqref{eq:dist_to_spec} is equivalent to showing that
  \begin{equation}\label{eq:bounded_boundaries}
    \abs{\psi_{-2n_0+1}}+\abs{\psi_{-2n_0+2}}+\abs{\psi_{2n_0-1}}+\abs{\psi_{2n_0}}<\frac12 \capconst(a)^{-n}.
  \end{equation}
   for some $n_0<n^\capconst$.
  Let $n_1=\capconst'n\lyap^\inv$ with $\capconst'=\frac9n\log2+3\log \capconst(a)$. If for some $0<m<n^\capconst$
  \begin{equation}\label{eq:lyap_est_1}
    \absbig{\frac1{n_1}\log\norm{\trans_{n_1}(m\EF,z)\cdots\trans_{1}(m\EF,z)}-\lyap_{n_1}}<\kappa \capconst,
  \end{equation}
  it follows from Proposition \ref{prop:exp_decay} that
  \begin{equation*}
    \abs{R^{[m,n_1+m]}_{\EF,z}(2x-i,2y-j)}<e^{-\abs{x-y}\lyap+\capconst\epsilon n_1}.
  \end{equation*}
  This implies for $k=m+[\frac{n_1}2]$ ($[x]$ being the nearest integer to $x$) that
  \begin{align*}
    \abs{\psi_k}   &< n^\capconst\abs{R^{[m,n_1+m]}_{\EF,z}(k,2j+1)}+n^\capconst\abs{R^{[m,n_1+m]}_{\EF,z}(k,2(n_1+j))} \\
                    &< n^\capconst e^{-\frac{n_1}2\lyap+\capconst\epsilon n_1}    \\
                    &< e^{-\frac{n_1}3\lyap} \\
                    &<\frac18\capconst(a)^{-n}
  \end{align*}
  by definition of $n_1$, and the same estimate holds for $\abs{\psi_{k-1}}$. Similar reasoning leads to $\abs{\psi_{-k}},\abs{\psi_{-k+1}}<\capconst(a)^{-n}/8$ and therefore to \eqref{eq:bounded_boundaries} if
  \begin{equation}\label{eq:lyap_est_2}
    \absbig{\frac1{n_1}\log\norm{\trans_{n_1}(-(m+n_1)\EF,z)\cdots\trans_{1}(-(m+n_1)\EF,z)}-\lyap_{n_1}}<\kappa \capconst.
  \end{equation}
  To show \eqref{eq:lyap_est_1}, \eqref{eq:lyap_est_2} we invoke Lemma \ref{lem:large_averages}, which for $M=n^\capconst$ gives
  \begin{equation*}
    \begin{split}
        \frac1{n_1M}\sum_{m=M+1}^{2M}\left(\log\norm{\trans_{n_1}(m\EF,z)\cdots\trans_{1}(m\EF,z)}+\log\norm{\trans_{n_1}(-(m+n_1)\EF,z)\cdots\trans_{1}(-(m+n_1)\EF,z)}\right) \\
            =2\lyap_{n_1}+\mathcal O(n_1^\inv).
    \end{split}
  \end{equation*}
  Thus, there exists $M<m\leq 2M$ such that
  \begin{equation*}
    \begin{split}
        \frac1{n_1}\left(\log\norm{\trans_{n_1}(m\EF,z)\cdots\trans_{1}(m\EF,z)}+\log\norm{\trans_{n_1}(-(m+n_1)\EF,z)\cdots\trans_{1}(-(m+n_1)\EF,z)}\right) \\
            >2\lyap_{n_1}+\mathcal O(n_1^\inv).
    \end{split}
  \end{equation*}
  Using Lemma \ref{lem:lyap_upper_bound} this implies
  \begin{align*}
    \lyap_{n_1}+\kappa \capconst>\frac1{n_1}\log\norm{\trans_{n_1}(m\EF,z)\cdots\trans_{1}(m\EF,z)}>2\lyap_{n_1}+\mathcal O(n_1^\inv)-(\lyap_{n_1}+\kappa \capconst),
  \end{align*}
  i.e. \eqref{eq:lyap_est_1}. Analogously, we prove \eqref{eq:lyap_est_2}.
\end{proof}


\section*{Acknowledgements}
The authors thank Jake Fillman for  clarifications of the literature on localization proofs for quasi-periodic Schr\"{o}dinger operators and critical reading of this manuscript, and Luis Velázquez for illuminating discussions about the connection between squares of quantum walks and CMV matrices.

C. Cedzich acknowledges support by the projet PIA-GDN/QuantEx P163746-484124 and by {\em DGE -- Ministère de l'Industrie}.

A. H. Werner thanks the Villum Fonden for its support via a Villum Young Investigator Grant. 

\bibliographystyle{abbrvArXiv}

\bibliography{ewalklocbib}

\end{document}